\UseRawInputEncoding
\documentclass[aps,showpacs,preprintnumbers,twocolumn]{revtex4}

\usepackage[centertags]{amsmath}
\usepackage{amsfonts}
\usepackage{amssymb}
\usepackage{amsthm}
\usepackage{newlfont}
\usepackage{graphicx}
\input{epsfx}

\newtheorem{thm}{Theorem}[section]
\newtheorem{lem}[thm]{Lemma}

\newtheorem{defn}[thm]{Definition}

\begin{document}

\title{Are temporal quantum correlations generally non-monogamous?}
\author{Marcin Nowakowski\footnote{Electronic address: marcin.nowakowski@pg.edu.pl}}
\affiliation{Faculty of Applied Physics and Mathematics,
~Gdansk University of Technology, 80-952 Gdansk, Poland}
\affiliation{National Quantum Information Center of Gdansk, Andersa 27, 81-824 Sopot, Poland}

\pacs{03.67.-a, 03.67.Hk}

\begin{abstract}
In this paper we focus on the underlying quantum structure of temporal correlations and show their peculiar nature which differentiate them from spatial quantum correlations. We show rigorously that a particular entangled history, which can be associated with a quantum propagator, is monogamous to conserve its consistency throughout time. Yet evolving systems violate monogamous Bell-like multi-time inequalities. This dichotomy, being a novel feature of temporal correlations, has its roots in the measurement process itself which is discussed by means of the bundles of entangled histories. We introduce and discuss a concept of a probabilistic mixture of quantum processes by means of which we clarify why the spatial-like Bell-type monogamous inequalities are further violated.
We prove that Tsirelson bound on temporal Bell-like inequalities can be derived from the entangled histories approach and as a generalization, we derive the quantum bound for multi-time Bell-like inequalities.
It is also pointed out that what mimics violation of monogamy of temporal entanglement is actually just a kind of polyamory in time but monogamy of entanglement for a particular evolution still holds.

\end{abstract}

\maketitle

\section{Introduction}

Recent years have proved a great interest of quantum entanglement monogamy concept showing its usability in quantum communication theory and its applications to quantum secure key generation \cite{Lutk1, Lutk2, Lutk3, Lutk4, Devetak05, KLi}. While spatial quantum correlations and especially their non-locality became a central subject of quantum information theory and their applications to quantum computation, potentiality of application of temporal non-local correlations is poorly analyzed. Yet there is a growing interest which is related to better understanding of this peculiar quantum phenomenon. The crucial issue relates to the very nature of time and temporal correlations phenomenon with their understanding within the framework of modern quantum and relativistic theories.

Non-local nature of quantum correlations in space has been accepted as a consequence of violation of local realism, expressed in Bell's theorem \cite{Bell} and analyzed in many experiments \cite{Aspect, Freedman}. As an analogy for a temporal domain, violation of macro-realism \cite{LGI2} and Leggett-Garg inequalities \cite{LGI} seem to indicate non-local effects in time and are a subject of many experimental considerations \cite{Chu, EX1, EX2, EX3, EX4}.
There have been different formalisms proposed for study of quantum temporal correlations including Multiple-Time States (MTS) by Aharonov {\it et al.} \cite{AAD,MTS1} as part of the Two-State-Vector formalism (TSVF) \cite{ABL,Properties,TSVFR,TTI}, the Entangled Histories (EH) approach  \cite{Cot} or the pseudo-density operators (PDOs) \cite{Vedral}.

The TSVF led to surprising effects within pre- and postselected systems (e.g. \cite{Paradoxes,Pigeon,Dis}), time travel thorough post-selected teleportation \cite{Lloyd1,Lloyd2}, a novel notion of quantum time \cite{Each}, new results regarding quantum state tomography \cite{MTS2} and a better understanding of processes with indefinite causal order \cite{NewSandu}, while the Entangled Histories approach led to Bell tests for histories \cite{WC3} and have been recently used for analysis of the final state proposal in black holes \cite{CN}. The subject of the black hole information loss paradox has been also addressed with application of PDOs \cite{Marletto} but engaging a concept of non-monogamy of spatio-temporal correlations.

In this paper we study the nature of the non-mongamous behavior of temporal correlations both for ensembles of quantum processes and their single instances. We show rigorously that a particular entangled history, which can be associated with a quantum propagator, is monogamous to  conserve  its  consistency  throughout  time.   Yet  evolving  systems  violate  monogamous  Bell-like multi-time  inequalities which can be explained engaging bundles of histories with the same pre-selected and post-selected states as initial and final boundaries for the considered evolution. This dichotomy does not have a counterpart in spatial domain and as such is a novel feature of temporal non-locality but is also a sign of importance of the internal structure of single processes. In particular, we prove also that the Tsirelson bound \cite{Tsirelson} on temporal Bell-like inequalities can be derived  from the entangled histories approach analytically and as a generalization, we derive the quantum bound for multi-time Bell-like inequalities which is not accessible by spatial quantum correlations. Furthermore,  it is also discussed that what mimics violation of monogamy of temporal entanglement  is  actually  just  a  kind  of  polyamory  in  time  but  monogamy  of  entanglement  for  a particular  evolution  still  holds which might influence the discussion about the resolution of the black-hole information paradox.

It is also crucial to emphasize that due to the isomorphism that can be derived for the TSVF and the entangled histories representations \cite{NowakowskiCohen}, the results presented in this paper can be achieved also for the TSVF which is also discussed partially in this paper.


\section{Review of Entangled Histories and Multiple-Time States}

Let us review briefly the entangled histories (EH) formalism and the multiple-time states (MTS) formalism as a natural extension of the two-state vector formalism (TSVF).

The predecessor of the entangled histories is the decoherent histories approach built on the grounds of the well-known Feynman's path integral theory for calculation of probability amplitudes of quantum processes. The EH formalism extends the concepts of the consistent histories theory by allowing for complex superposition of histories. A history state is understood as an element in $\text{Proj}(\mathcal{H})$, spanned by projection operators from $\mathcal{H}$ to $\mathcal{H}$, where  $\mathcal{H}=\mathcal{H}_{t_{n}}\odot...\odot\mathcal{H}_{t_{1}}$.
The $\odot$ symbol, which we use to comply with the current literature, stands for sequential tensor products, and has the same meaning as the above $\otimes$ symbol.
The alternatives at a given instance of time form an exhaustive orthogonal set of projectors  $\sum_{\alpha_{x}}P_{x}^{\alpha_{x}}=\mathbb{I}$ and for the sample space of entangled histories $|H^{\overline{\alpha}})=P_{n}^{\alpha_{n}}\odot P_{n-1}^{\alpha_{n-1}}\odot\ldots\odot P_{1}^{\alpha_{1}}\odot P_{0}^{\alpha_{0}}$ ($\overline{\alpha}=(\alpha_{n}, \alpha_{n-1},\ldots, \alpha_{0})$), there exists $c_{\overline{\alpha}} \in \mathbb{C}$ such that $\sum_{\overline{\alpha}}c_{\overline{\alpha}}|H^{\overline{\alpha}})=\mathbb{I}$.

As an example, one can take a history $|H)=[z^+]\odot[x^-]\odot[y^-]\odot[x^+]=[|z^+\rangle\langle z^+|]\odot[|x^-\rangle\langle x^-|]\odot[|y^-\rangle\langle y^-|]\odot[|x^+\rangle\langle x^+|]$ for a spin-$\frac{1}{2}$ particle being in an eigenstate of the Pauli-X operator at time $t_1$, in an eigenstate of the Pauli-Y operator at time $t_2$, and so on. Within this formalism one also defines the unitary bridging operators $\mathcal{T}(t_j,t_i):\mathcal{H}_{t_i}\rightarrow\mathcal{H}_{t_j}$ evolving the states between instances of time, and having the following properties: $\mathcal{T}(t_j,t_i)=\mathcal{T}^{\dagger}(t_i,t_j)$ and $\mathcal{T}(t_j,t_i)=\mathcal{T}(t_j,t_{j-1})\mathcal{T}(t_{j-1},t_i)$.
This formalism introduces also the chain operator $K(|H^{\overline{\alpha}}))$, which can be directly associated with a time propagator of a given quantum process:
\begin{equation}
K(|H^{\overline{\alpha}}))=P_{n}^{\alpha_n}\mathcal{T}(t_{n},t_{n-1}) P_{n-1}^{\alpha_{n-1}}\ldots P_{1}^{\alpha_1}\mathcal{T}(t_{1},t_{0})P_{0}^{\alpha_0}
\end{equation}
This operator plays a fundamental role in measuring a weight of any history $|H^{\alpha})$:
\begin{equation}
W(|H^{\alpha}))=TrK(|H^{\alpha}))^{\dagger}K(|H^{\alpha}))
\end{equation}
which can be interpreted as a realization probability of a history by the Born rule application. The histories approach requires also that the family of histories is consistent, i.e. one can associate with a union of histories a weight equal to the sum of weights
associated with particular histories included in the union.

Multiple-Time States (MTS) extend the standard quantum mechanical state by allowing its simultaneous description in several different moments. Such a multiple-time state may encompass both forward- and backward-evolving states on equal footing. MTS represent all instances of collapse (i.e. those moments in time when the quantum state coincided with an eigenstate of some measured operator) and allow them to evolve both forward and backward in time. This evolution backwards in time can be understood literally (giving rise to the Two-Time Interpretation \cite{TTI}), but this is not necessary, it can be simply regarded as a mathematical feature of the formalism (which is, in fact, equivalent to the standard quantum formalism \cite{TSVFR}). MTS live in a tensor product of Hilbert spaces $\mathcal{H}$ admissible at those various instances of time ($t_1<...<t_n$) denoted by \cite{MTS1}
\begin{equation}\label{HMTF}
\mathcal{H}=\mathcal{H}_{t_{n}}^{(\cdot)}\otimes...\otimes\mathcal{H}_{t_{k+1}}^{\dagger}\otimes\mathcal{H}_{t_{k}}\otimes\mathcal{H}_{t_{k-1}}^{\dagger}\otimes...\otimes\mathcal{H}_{t_{1}}^{(\cdot)},
\end{equation}
where a dagger means the corresponding Hilbert space consists of states which evolve backwards in time. The initial and final Hilbert spaces might be daggered or not (this is denoted by a ``$\cdot$'' superscript). All Hilbert spaces containing either (forward-evolving) kets or (backward-evolving) bras are alternating to allow a time-symmetric description at any intermediate moment.

As an example of (a separable) MTS we can consider the following state: $_{t_4}\langle z^+||x^-\rangle_{t_3~t_2}\langle y^-||x^+\rangle_{t_1} \in \mathcal{H}_{t_{4}}^\dagger\otimes\mathcal{H}_{t_{3}}\otimes\mathcal{H}_{t_{2}}^\dagger\otimes\mathcal{H}_{t_{1}}$. This multiple-time state represents an initial eigenstate of the Pauli-X operator evolving forward in time from $t_1$ until collapse into an eigenstate of the Pauli-Y operator occurs at time $t_2$. Later on, at time $t_3$ the system is projected again onto a different eigenstate of the Pauli-X operator. Finally at $t_4$ the system is measured in the Z basis, and the resulting eigenstate evolves backward in time. In the following we will focus on two-time states (sometimes called two-states), which consist of a forward evolving state $|\psi_1\rangle_{t_{1}}$ and a backward evolving state $|\psi_2\rangle_{t_{2}}$ in the above form $_{t_2}\langle \psi_2| |\psi_1\rangle_{t_1}$ to achieve a richer description of a quantum system during the time interval $t_1\le t \le t_2$ \cite{TSVFR}.

Given an initial state $|\Psi\rangle$ and a final state $\langle \Phi|$, the probability that an intermediate measurement of some hermitian operator $A$ will result in the eigenvalue $a_n$ is given by the ABL formula \cite{ABL}

\begin{equation} \label{ABL}
p(A=a_n)=\frac{1}{N}|\langle\Phi|U_{2}P_{n}U_{1}|\Psi\rangle|^2,
\end{equation}
where $U_i$ represent unitary evolution, the operator $P_n$ projects on $|a_n\rangle$ and
\begin{equation}
N \equiv \sum_k |\langle\Phi|U_{2}P_{k}U_{1}|\Psi\rangle|^2.
\end{equation}
This probability rule is important in that it uses the information available through the final state in a way which is manifestly time-symmetric.

\section{Monogamy of a particular quantum process}

The fundamental property of spatial quantum entanglement is its monogamy. This property states that for a tripartite system ABC, maximal entanglement of the pair AB excludes its non-local correlations with the third party, i.e. if $\rho_{AB}=|\Psi^+\rangle\langle\Psi^+|$, then any extension of this state is of the form $\rho_{ABC}=|\Psi^+\rangle\langle\Psi^+|\otimes |\Psi\rangle\langle\Psi|$. For the temporal correlations, it seems that this property does not hold, especially when one considers statistical distribution of measurement results \cite{White, Marletto}. Yet, what is obvious in the spatial case does not have mere analogies in the temporal case. We will show now that a particular history can be monogamous but further we will discuss how temporal correlations can lead to non-monogamous results for bundles of histories with which we tackle during the measurement process. This subtlety is rather a sign of a deeper nature of quantum processes which can keep their consistency for particular instances, yet leads do quite counter-intuitive results for their ensembles.

Suppose we have two non-equivalent multi-time entangled histories of an evolving qubit through times ${t_4> t_3> t_2> t_1}$ for which we consider the past effect of the measurement at time $t_4$:

\begin{eqnarray}
    |H_1)&=&\frac{1}{\sqrt{2}}[|0)\odot|0)\odot|0)\odot|0)+|1)\odot|1)\odot|1)\odot|1)] \nonumber \\
    |H_2)&=&\frac{1}{\sqrt{2}}|0)\odot[|0)\odot|0)\odot|0)+|1)\odot|1)\odot|1)]
\end{eqnarray}

The history $|H_1)$ can be perceived as a superposition of two histories on times ${t_4, t_3, t_2, t_1}$. If one measures this evolution at time $t_4$ with dichotomic projective observables $P_0=|0\rangle\langle 0|$ and $P_1=|1\rangle\langle 1|$, we can conclude that the state was with probability $p_0=\frac{1}{2}$ in a history $|H_{10})=|0)\odot|0)\odot|0)$ at previous times and with probability $p_1=\frac{1}{2}$ in a history $|H_{11})=|1)\odot|1)\odot|1)$. Alternatively, one can consider an ensemble of history states $\{\{p_0, |H_{10})\},\{p_1, |H_{11})\}\}$, i.e. half of the qubits evolving trivially in a history $|H_{10})$ and half in $|H_{11})$ through times ${t_3, t_2, t_1}$ which can be represented by a history super-operator $\rho_H=\frac{1}{2}(|H_{10})(H_{10}| + |H_{11})(H_{11}|)$.
This evolution is different for the history $|H_2)$. If one performs the same measurements at time $t_4$, then we get an entangled history through times ${t_3, t_2, t_1}$ for the projective measurement $P_0$ at time $t_4$. Thus, physically we can propose the concept of \textit{the probabilistic mixture of histories}:

\begin{defn}
A mixed history state is defined as a positive super-operator acting on a history state space:
\begin{equation}
\rho_{hist}=\sum_{i}p_i |H_i)(H_i|
\end{equation}
where $Tr\rho_{hist}=1$, $\sum_{i} p_i=1$ and $\forall_i 1>p_i\geq 0 $.
\end{defn}
This mixture of histories can be naturally associated with an ensemble of histories $\{p_i, |H_i)\}$.
Following, we consider an example of a spin particle traversing two paths to check a future influence of the measurement at time $t_1$:

\textit{Example 1. }Imagine a spin-$\frac{1}{2}$ particle at three times $\{t_3, t_2, t_1\}$ evolving trivially by $\mathcal{T}=\mathbb{I}$ with a family of entangled histories:
\begin{eqnarray}
|H^1)&=&\sqrt{2}([z^{+}]\odot[x^{+}]\odot[z^{+}]+[z^{-}]\odot[x^{-}]\odot[z^{+}])\nonumber\\
|H^2)&=&\sqrt{2}([z^{-}]\odot[x^{+}]\odot[z^{+}]+[z^{+}]\odot[x^{-}]\odot[z^{+}])\nonumber\\
|H^3)&=&\sqrt{2}([z^{+}]\odot[x^{+}]\odot[z^{-}]+[z^{-}]\odot[x^{-}]\odot[z^{-}])\nonumber\\
|H^4)&=&\sqrt{2}([z^{-}]\odot[x^{+}]\odot[z^{-}]+[z^{+}]\odot[x^{-}]\odot[z^{-}])\nonumber\\
\end{eqnarray}
If we consider a state $|\Phi)=\frac{1}{\sqrt{2}}|H^1)+\frac{1}{\sqrt{2}}|H^2)$, then a particle, measured at time $t_1$ and having a spin up in a direction $z^{+}$, can evolve within the history $|H^1)$ with probability $P(|H^1))=\frac{1}{2}$ and be in the history $|H^2)$ with probability $P(|H^2))=\frac{1}{2}$.\\
Noteworthily, one can also find in the space of histories $\mathcal{S}=span\{|H^1), |H^2), |H^3), |H^4)\}$ the following temporal GHZ-like vector \cite{Cot} (normalized for $|\alpha|^2+|\beta|^2=1$):
\begin{eqnarray}
|\tau GHZ)&=&\frac{\alpha}{\sqrt{2}}|H^1)+\frac{\alpha}{\sqrt{2}}|H^2)+\frac{\beta}{\sqrt{2}}|H^3)+\frac{\beta}{\sqrt{2}}|H^4)\nonumber\\
&=&\alpha[z^{+}]\odot[z^{+}]\odot[z^{+}]+\beta[z^{-}]\odot[z^{-}]\odot[z^{-}] \nonumber\\
\end{eqnarray}


Having grounded the concept of a mixture of histories, to better understand the behavior of entangled histories in a context of monogamy, let us consider the Mach-Zehnder interferometer (Fig.\ref{Mach-Zender Interferometer}) and potential superposed evolutions of a photon state in the context of monogamy problem of temporal entanglement.

There are many potential histories allowed in the interferometer which can be also steered by selection of the external detection points but let us consider firstly the one displaying quantum entanglement in time at times $\{t_2, t_1\}$:
\begin{equation}
    |H)=\alpha|\phi_{3,2})\odot(|\phi_{2,1})\odot |\phi_{1,1})+|\phi_{2,2})\odot |\phi_{1,2}))\odot |\phi_0)
\end{equation}
This particular history can be realized by placement of a detector at time $t_3$ which detects a photon in state $\rho=|\phi_{3,2}\rangle\langle \phi_{3,2}|$ and displays quantum entanglement in time for times $\{t_2, t_1\}$:
\begin{equation}
 |H_{t_2, t_1})=\alpha(|\phi_{2,1})\odot |\phi_{1,1})+|\phi_{2,2})\odot |\phi_{1,2}))
\end{equation}

Interestingly, this entangled history at times $\{t_2, t_1\}$ cannot be derived from the following temporal version of a GHZ-state which also can be realized in this interferometer:
\begin{equation}
|\widetilde{H})=\alpha(|\phi_{3,1})\odot |\phi_{2,1})\odot |\phi_{1,1})+|\phi_{3,2})\odot |\phi_{2,2})\odot |\phi_{1,2}))\odot |\phi_0)
\end{equation}

Let us observe that the reduced component of this history $|\phi_{3,1})\odot |\phi_{1,1})$ is correlated with $|\phi_{2,1})$ and not with $|\phi_{2,2})$. Thus, reduction of $|\widetilde{H})$ over times $t_2$ and $t_0$ is not a complex superposition of histories but is \textit{a probabilistic mixture} as already stated in this section:
\begin{equation}
\rho_{t_1t_3}=\alpha\alpha^*(|\phi_{3,1}\phi_{1,1})(\phi_{3,1}\phi_{1,1}|+|\phi_{3,2}\phi_{1,2})(\phi_{3,2}\phi_{1,2}|)
\end{equation}
where $|\phi_{3,1}\phi_{1,1})=|\phi_{3,1})\odot |\phi_{1,1})$ etc.
This can be also formally derived employing a temporal  partial trace operator \cite{TraceTime} over time instances:
$\rho_{t_1t_3}=Tr_{t_2t_0}|\widetilde{H})(\widetilde{H}|$. This operator is an analogy of spatial tracing out but has to keep consistency of the evolution - a condition which is not present in spatial case.

On the contrary, if this reduction of the temporal GHZ-state $|\tau GHZ)$ would be a complex superposition of entangled histories, i.e. it could be always expanded to a history of the following type e.g. $|\varphi_{t_{x}})\odot(|\varphi_{3,1})\odot |\varphi_{1,1})+|\varphi_{3,2})\odot |\varphi_{1,2}))$ or the aforementioned $|H)$ which implies decorrelation with the next instance of the history in such a case by employing e.g. a projective measurement. We should emphasize that this reasoning is in agreement with the Feynman's addition rule for probability amplitudes.

It is important
to note that these considerations are related to $|H)(H|$ - observable and the particular history $|H)$. Yet, other histories in the Mach-Zehnder interferometer are also accessible. It shows clearly a physical sense
of quantum entanglement in time and further a concept of its monogamy for a particular entangled history.

\begin{figure}[h]
\centerline{\includegraphics[width=8cm]{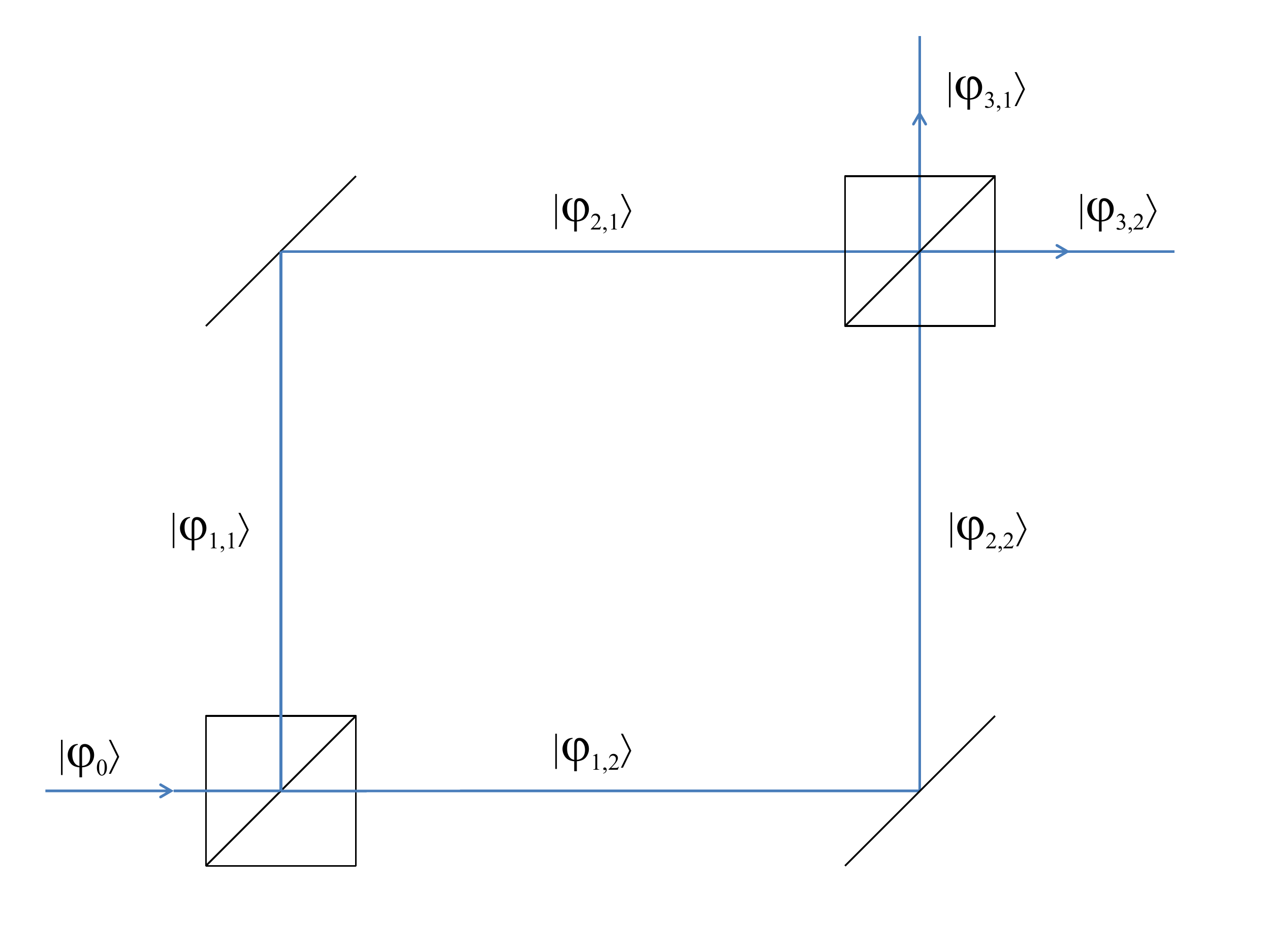}}
 \caption[Mach-Zehnder Interferometer]{The Mach-Zehnder interferometer with an input state $|\varphi_{0}\rangle$ - a vacuum state is omitted which does not change further considerations. The beam-splitters can be represented by Hadamard operation
 acting on the spatial modes.
 }
 \label{Mach-Zender Interferometer}
\end{figure}

We can summary these considerations with the following lemma about temporal entangled histories:

\begin{lem}
There does not exist  any such a history $|H_{ABC})$ acting on space $\mathcal{H}_{ABC}=\mathcal{H}_{t_3}\odot\mathcal{H}_{t_2}\odot\mathcal{H}_{t_1}$ at times $\{t_{3},t_{2},t_{1}\}$ so that one can find reduced histories $|H_{AB})=\frac{1}{\sqrt{2}}(|e_{0})\odot|e_{0})+|e_{1})\odot|e_{1}))$ and $|H_{BC})=\frac{1}{\sqrt{2}}(|e_{0})\odot|e_{0})+|e_{1})\odot|e_{1}))$.
\end{lem}
\begin{proof}
The assumption that $|H_{AB})=|H_{BC})$ implies immediately dimension of the history $|H_{ABC})$. To prove this lemma, we apply the formalism of Feynman propagators.
As already mentioned, the particular history $|H_{ABC})$ of a physical process $A B C$ can be associated with a transition amplitude $\Gamma_{ABC}$
($\Gamma_{ABC} \sim e^{iS_{ABC}}$ where $S_{ABC}$ stands for an action functional of the process) by means of the $K$-operator:
\begin{eqnarray}
\Gamma_{ABC}&=&\sum_{ijk}\alpha_{ijk}\langle e_{i}|e^{iS_{ij}^{AB}}|e_{j}\rangle\langle e_j|e^{iS_{jk}^{BC}}|e_k\rangle \; (\alpha_{ijk}\in \mathcal{C})\nonumber \\
&=&\sum_{ijk}\alpha_{ijk} \Gamma_{ijk}\nonumber
\end{eqnarray}
Note that the assumed dimension of the subsystems implies the choice of the basis $\{|e_{i}\rangle\}$ for this propagator.
Since we consider a history $|H_{ABC})$ with reduced histories (for sub-processes) $|H_{AB})=|H_{BC})=\frac{1}{\sqrt{2}}(|e_{0})\odot|e_{0})+|e_{1})\odot|e_{1}))$, one can associate with them the following transition probability amplitudes:
\begin{equation}
\Gamma_{AB}=\Gamma_{BC}=\alpha_{00}\Gamma_{00}+\alpha_{11}\Gamma_{11}
\end{equation}
and the existence of summands $\Gamma_{00}$ and $\Gamma_{11}$ implies that:
\begin{equation}
\Gamma_{ABC}=\sum_{i}(\alpha_{i00}\Gamma_{i00}+\alpha_{i11}\Gamma_{i11})=\sum_{j}(\alpha_{00j}\Gamma_{00j}+\alpha_{11j}\Gamma_{11j})
\end{equation}
On the contrary, if there exists a component e.g. $\Gamma_{i01}$ in the propagator $\Gamma_{ABC}$, the history $|H_{i01})$ could be also realized. Yet it is not the case for the assumed global history $|H_{ABC})$ for which we do not consider sub-histories $|e_0)\odot|e_1)$. Thus, only the following elements can contribute to the global propagator and corresponding history:
\begin{equation}
\Gamma_{ABC}=\alpha_{000}\Gamma_{000}+\alpha_{111}\Gamma_{111}
\end{equation}
However, this propagator, as already discussed, could be associated with a temporal version of GHZ-like state: $|\tau GHZ)=\alpha_{000} |0)\odot|0)\odot|0) +  \alpha_{111}|1)\odot|1)\odot|1)$ (we assume that the coefficients are non-zero) from which we cannot derive reduced entangled history, thus, contradicting the structure of the global history $|H_{ABC})$.
\end{proof}

As a consequence, it is natural to conclude that \textit{any Feynman propagator is also monogamous}. This lemma holds naturally also for any finite dimension n of the system $A$ evolving in time.

A natural consequence of entanglement monogamy in space is that 
we cannot build a quantum spatial state where a chosen party is entangled with an infinite number of parties. In principle, if we consider Feynman path integral which integrates all probability amplitudes over possible evolution paths between two space-time points, one can state a question about correlations between a state of a system at a chosen time $t_{x}$ and all other times separated by $dt$ in this
evolution. Suppose that we are considering a two-state history $|F)=[x_{E}]\odot [x_{S}]$ where a particle is localized at $x_{S}$ at time $t_{S}$ (our initial state is $|S\rangle=|x_{S}\rangle$) and evolves to the final state $|E\rangle$ localized at $x_{E}$ at time $t_{E}$. This history can be further expanded
as a Feynman path integral \cite{Schwartz, Feynman} assuming breaking down the evolution time into $n$ small time intervals $\delta t$:
\begin{eqnarray}
\langle E|S\rangle &=&\int \mathcal{D}x \exp\{\frac{i}{\hslash} \int_{x_S, t_S}^{x_E, t_E}dt\mathcal{L} \} \\
&=&\int dx_{n}\cdots dx_{1} \langle x_{E}| e^{-iH(t_{E})\delta t}|x_{n}\rangle\langle x_{n}|\cdots \\ \nonumber
 && |x_{2}\rangle\langle x_{2}| e^{-iH(t_{2})\delta t}|x_{1}\rangle\langle x_{1}|e^{-iH(t_{1})\delta t} |x_{S}\rangle \nonumber
\end{eqnarray}
where we assumed evolution steered by hamiltonian $H$ being a smooth function of $t$.

This evolution can be represented as a two-time history with an initial pre-selected and final well-defined state: $|F)=[x_{E}]\odot [x_{S}]$ for times $t_{S}$ and $t_{E}$. However, what is substantial in this consideration, all intermediate times are undefined for the external observer of the evolution so this particular history is separable for an observer pre-selecting the state $[x_{S}]$ and post-selecting $[x_{E}]$. In a consequence, asking for correlations of the state at time e.g. $t_{S}$ and say $t_{S}+\delta t$ does not make sense (unless $t_{E}=t_{S}+\delta t$).

We can represent the product $|E\rangle\langle E|S\rangle\langle S|$ by means of integration over histories (it is crucial to remember that in similarity to path integral summands not every
particular history summand is assumed to be consistent, i.e. physically realizable):
\begin{equation}
\begin{split}
|E\rangle\langle E|S\rangle\langle S|=\int dx_{n}\cdots dx_{1} K([x_{E}]\odot[x_{n}]\odot\ldots \\
\odot[x_{2}]\odot[x_{1}]\odot[x_{S}])
\end{split}
\end{equation}
It represents an expansion of a quantum propagator with quantum histories contracted by K-operator.

To deepen our understanding of differences between spatial and temporal correlations, let us reconsider a spatial and temporal version of $GHZ$-state from a resource perspective:
\begin{eqnarray}
    |\Psi_{ABC}\rangle&=&\frac{1}{\sqrt{2}}(|000\rangle + |111\rangle) \nonumber \\
    |\Psi_{t_2 t_1 t_0})&=&\frac{1}{2}([0]\odot[0]\odot[0]+[1]\odot[1]\odot[1])
\end{eqnarray}

While a three-qubit spatial state $|\Psi_{ABC}\rangle$ cannot be simply extended to $n$-qubit state due to lack of additional resources, a temporal state can be always extended to $n$-times (reaching infinity as a limit if we assume that time slicing does not have its limit like a Planck time) between the constrained past $t_{0}$ and future $t_{2}$ as far as the time steps are defined for the external observer:

\begin{eqnarray}
    |\Psi_{ABC}\rangle & \nrightarrow & |\Psi_{ABCC_1\ldots}\rangle\\
    |\Psi_{t_2 t_1 t_0})&\rightarrow &|\Psi_{t_2\ldots t_1 t_{0n}\ldots t_{01} t_0})
\end{eqnarray}

This tricky feature of temporal correlations is one of the reasons of polyamoric nature of time  as we will see in next paragraphs.

\section{Tsirelson's Bound from Entangled Histories and General Bound for temporal correlations}

The violation of local realism (LR) \cite{Bell} and macrorealism (MR) \cite{LGI2} by quantum theories has been studied for many years in experimental setups where measurements' data are tested against violation of Bell inequalities for LR and Leggett-Garg inequalities (LGI) \cite{LGI} for MR. For quantum theories, the former raises as a consequence of non-classical correlations in space while the latter as a consequence of non-classicality of dynamic
evolution. In this section we show that entangled histories approach gives the same well-known Tsirelson bound \cite{Tsirelson} on quantum correlations for LGI as quantum entangled states in case of bi-partite spatial correlations for CHSH-inequalities and  derive a general quantum bound on multi-time Bell-like inequalities.

In the temporal version of CHSH-inequality being a modification of original Leggett-Garg inequalities, Alice performs measurement at time $t_{1}$ choosing between two dichotomic observables $\{A_{1}^{(1)}, A_{2}^{(1)}\}$ and then Bob performs a measurement at time $t_{2}$ choosing between $\{B_{1}^{(2)}, B_{2}^{(2)}\}$. Therefore, the structure of this LGI can be represented as follows \cite{Vedral}:
\begin{equation}
S_{LGI}\equiv c_{12}+c_{21}+c_{11}-c_{22} \leq 2
\end{equation}
where $c_{ij}=\langle A_{i}^{(1)}, B_{j}^{(2)}  \rangle$ stands for the expectation value of consecutive measurements performed at time $t_{1}$ and $t_{2}$.

Since one can build in a natural way $\mathcal{C}^{*}$-Algebra of history operators for normalized histories from projective Hilbert spaces equipped with a well-defined inner product, we provide reasoning about bounding the LGI purely on the space of entangled histories and achieve the quantum bound $2\sqrt 2$ of CHSH-inequality specific for spatial correlations. The importance of this result achieved analytically is due to the fact that
previously it was derived basing on convex optimization methods by means of semi-definite programming \cite{Budroni} and by means of correlator spaces \cite{Fritz} not being equivalent to probability space (probability conditional distributions
of consecutive events) without underpinning mathematical structure of quantum temporal states.

In a temporal setup one considers measurements $\mathbb{A}=I \odot \mathbb{A}^{(1)}$ (measurement $\mathbb{A}$ occurring at time $t_{1}$) and $\mathbb{B}=\mathbb{B}^{(2)}\odot I$ which is in an exact analogy to the proof of the above theorem for a spatial setup. The history with 'injected' measurements can be represented as $|\widetilde{H})=\alpha \mathbb{A}\mathbb{B}|H)\mathbb{A}^{\dagger}\mathbb{B}^{\dagger}$ where $\alpha$ stands for a normalization factor.
History observables are history state operators which are naturally Hermitian and their eigenvectors can generate a consistent history family\cite{Cot}. \\
As an example, we can consider  spin $\frac{1}{2}$-particle with a history inducing evolution $|\psi(t_1)\rangle\rightarrow|\psi(t_2)\rangle$ on which we act with $\sigma_y\odot\sigma_x$ operation. This step results with a new effective history:
\begin{equation}
|\widetilde{H})=\alpha \sigma_y[\psi(t_2)]\sigma_y^{\dagger} \odot \sigma_x[\psi(t_1)]\sigma_x^{\dagger}
\end{equation}
For an observable $A=\sum_{i}a_{i}|H_{i})(H_{i}|$, its measurement on a history $|H)$ generates an expectation value $\langle A\rangle=Tr(A|H)(H|)$ (i.e. the result $a_{i}$ is achieved with probability $|(H|H_{i})|^{2}$) in analogy to the spatial case. Thus, one achieves history $|\widetilde{H})$ as a realized
history with measurements and the expectation value of the history observable $\langle A \rangle$. It is worth mentioning that $|\widetilde{H})$ and $|H)$ are both compatible histories, i.e. related by a linear transformation. Equipped with the aforementioned findings about history observables, one can state now the following lemma:
\begin{lem}
For any history density matrix $W$ and Hermitian history dichotomic observables $A_{i}=I\odot A_{i}^{(1)}$ and $B_{j}= B_{j}^{(2)} \odot I$ where $i,j \in \{1,2\}$ the following bound holds:
\begin{eqnarray}
S_{LGI}&=&c_{11}+c_{12}+c_{21}-c_{22}\\
&=&Tr((A_{1}B_{1}+A_{1}B_{2}+A_{2}B_{1}-A_{2}B_{2})W)\nonumber \\
&\leq& 2\sqrt{2}
\nonumber
\end{eqnarray}
\end{lem}
\begin{proof}
The proof of this observation can be performed in similar to the spatial version of CHSH-Bell inequality under assumption that the states are represented by
entangled history states and for two possible measurements $\{A_{1}^{(1)}, A_{2}^{(1)}\}$ at time $t_{1}$ and two measurements $\{B_{1}^{(1)}, B_{2}^{(1)}\}$ at time $t_{2}$. These operators can be of dimension $2\times2$ meeting the condition $A_{i}^{2}=B_{j}^{2}=I$. Therefore, they can be interpreted as spin components along two different directions. In consequence, it is well-known that the above inequality is saturated for $2\sqrt{2}$ for a linear combination of tensor spin correlation that holds also for temporal correlations.
Additionally, one could also apply for this temporal inequality reasoning based on the following obvious finding \cite{Tsirelson} that holds also for the temporal scenario due to the structure of $\mathcal{C}^{*}$-Algebra of history operators with $\odot$-tensor operation:
\begin{eqnarray}
A_{1}B_{1}+A_{1}B_{2}+A_{2}B_{1}-A_{2}B_{2}&\leq&\\
\frac{1}{\sqrt{2}}(A_{1}^2+A_{2}^2+B_{1}^2+B_{2}^2)&\leq&
2\sqrt{2}I \nonumber
\end{eqnarray}

\end{proof}

For temporal correlations measurements can lead to counter-intuitive results which do not occur for spatial quantum resources. Let us reexamine the case of $GHZ$ states firstly shared as a spatial system of three entangled qubits among Alice, Bob and Charlie: $|\Psi_{ABC}\rangle=\frac{1}{\sqrt{2}}(|000\rangle + |111\rangle)$ which obviously leads to a separable state for any pair from this system, e.g. $\rho_{AB}=\frac{1}{2}(|00\rangle\langle 00|+ |11\rangle\langle 11|)$. Assume further that they can choose from dichotomic projective observables: $P_0=|0\rangle\langle 0|$ or $P_1=|1\rangle\langle 1|$, then in this multipartite case any pair cannot identify alone without the third party that they are part of the more complex entangled system. This is the core difference from the temporal analog of this state. \\
In the temporal case the situation is quite opposite for the temporal version of the GHZ-state which is a sign of qualitative difference between spatial and temporal resources. Alice, Bob and Charlie, having instances of the same system but at different times, can in each pair detect non-locality in time.

When we measure an average value of the aforementioned Bell-like temporal inequality:
\begin{equation}
\langle S_{LGI} \rangle=\langle A_{1}B_{1}+A_{1}B_{2}+A_{2}B_{1}-A_{2}B_{2} \rangle
\end{equation}
we consider an ensemble of systems from which each quarter is measured against the observables $A_{i}B_{j}$.  It is easy to observe that with a choice of observables:
$A_1=Z, A_2=(Z+X)/\sqrt{2}, B_1=Z, B_2=(Z-X)/\sqrt{2}$ we get effectively the average value: $\langle S_{LGI} \rangle= \sqrt{2}\langle XX+ZZ \rangle$. It is easy to observe that $\langle XX \rangle =\langle ZZ \rangle = 1$ and the Tsirelson maximum is saturated. However, what is important in this simple example is that consecutive measurements of both X and Z leaves the system in the same eigenstate for any number of time steps. As an immediate implications, one gets violation of monogamous Bell-like inequalities in space \cite{Pawlowski}:
\begin{equation}
    S_{AB}+S_{BC}\nleq 4
\end{equation}
since for the temporal tripartite system $ABC$ (B and C being instances of A at consecutive times) we get saturation for the AB pair and for the BC:
\begin{equation}
S_{\tau AB}+S_{\tau BC}=4\sqrt{2}
\end{equation}
This limit cannot be achieved by spatial correlations \cite{Pawlowski}.
The fundamental point about generation of such averaged Bell-like inequalities is that we operate actually with a bundle of different histories (formally a bundle of vectors from the consistent entangled history family set) starting with the same initial pre-selected state and finalizing with the same post-selected final state for the bundle but having different intermediary steps (in the above example with XXX and ZZZ quantum operations).

We can look at the problem of bounding temporal correlations also by prism of the two-state vector  formalism which is isomorphic to the entangled histories \cite{NowakowskiCohen}. The correlations can be described by the probabilistic boxes in non-signalling theory. The box is shared between two parties who give the input setting $\{x, y\}$ of the measuring devices and get the outputs $\{a, b\}$ with probability $p(ab|xy)$ being an entry of the join probability distribution matrix $P(ab|xy)=[p(ab|xy)]$. All entries of this matrix meet the non-negativity condition ($p(ab|xy)\geq 0$) and are normalized:$\forall_{x,y}\sum_{a,b}p(ab|xy)=1$ and the no-signalling condition imposed on the quantum correlations by the special relativity constraints: the marginals $p(a|x)$ and $p(b|y)$ are independent of settings y and x respectively, i.e. $\forall_{y,a,x}p(a|x)=\sum_b p(ab|xy)$ and $\forall_{x,b,y}p(b|y)=\sum_a p(ab|xy)$.
Then the Aharanov-Bergmann-Lebowitz (ABL) formula (\ref{ABL})  delivers a method for calculation of the measurements probability in between the initial time with the pre-selected state and the post-selected state at the final time of the analyzed quantum process.

In the case of series of X and Z measurements injected in the aforementioned histories considered in this section we get the following example for probability distribution with assumption that at times $t_1$ and $t_2$ the X observable is chosen and we get $|\uparrow_x\rangle$ results in both times:
\begin{equation}
p(\uparrow_x\uparrow_x|XX)=\frac{|\langle \Phi |\uparrow_x\rangle\langle\uparrow_x| |\uparrow_x\rangle\langle\uparrow_x| \Psi\rangle|^2}{\sum_{ab}p(ab|XX)}
\end{equation}
This is an operational method for generation of the whole probability distribution matrix. However, we should note that these experiments start with the same initial and final states but with different intermediate steps, thus, leading to a bundle of histories at times ${t_0, t_1, t_2, t_3}$ (Fig. \ref{GlobalHistory}):
\begin{equation}
    |H_{abxy}) \sim p(ab|XY)
\end{equation}
which can lead to violation of spatial quantum bounds on Bell-like inequalities.

\begin{figure}[h]
\centerline{\includegraphics[width=9.5cm]{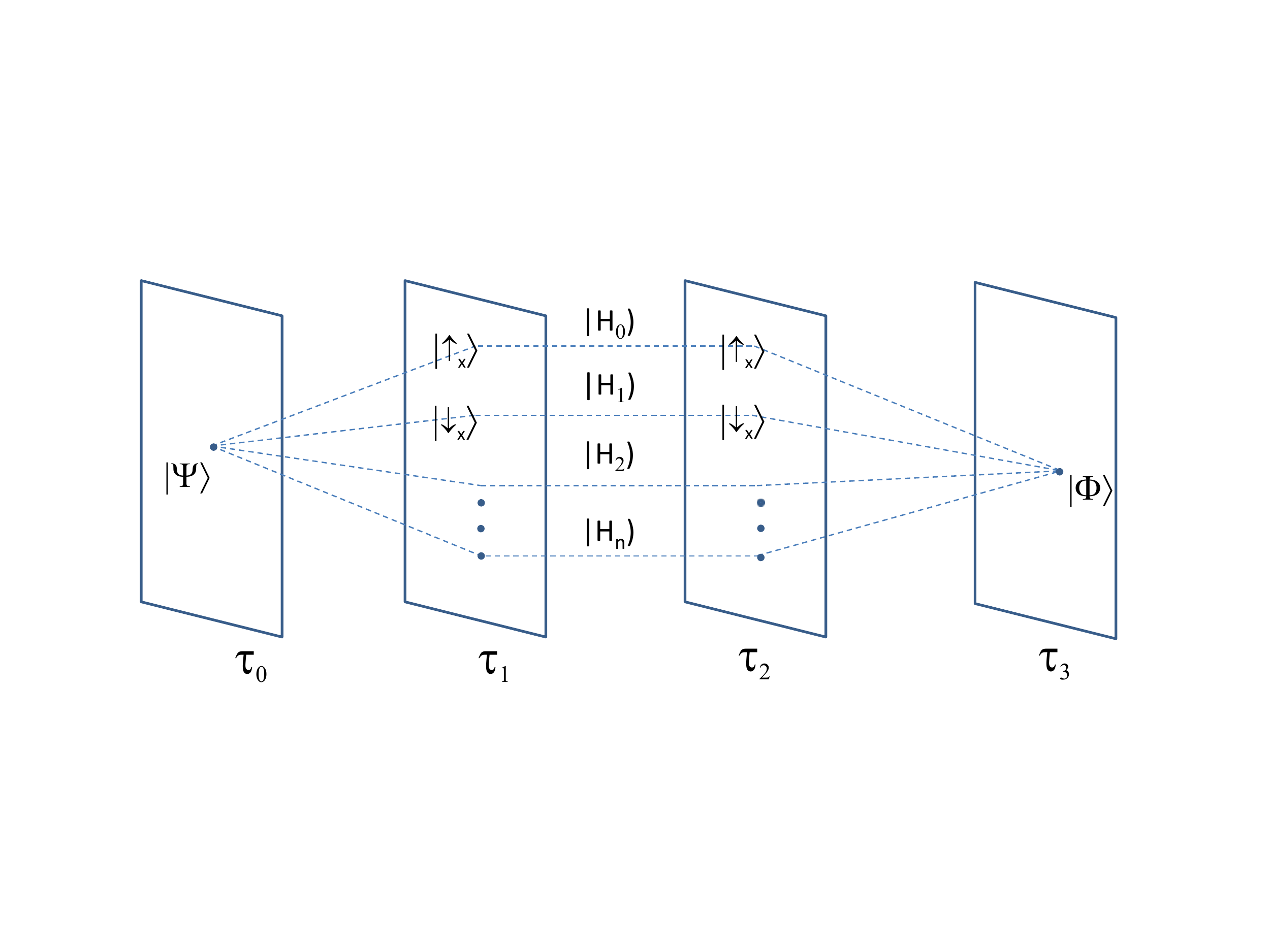}}
 \caption[GlobalHistory]{A bundle o histories at times ${\tau_0, \tau_1, \tau_2, \tau_3}$ with a pre-selected state $|\Psi\rangle$ and post-selected state $|\Phi\rangle$. Exemplary histories $|H_0)$ and $|H_1)$ with incorporated measurement results of X. This bundle contributes to violation of monogamous temporal Bell-like inequalities engaging different history states with the same initial and final states.}
 \label{GlobalHistory}
\end{figure}

We can then formulate generic bounds on temporal correlations of qubits in quantum theories (this result can be generalized to the qudits' case).

Let us assume that the quantum process occurs n times and that for any two times $\{t_i, t_{i+1}\}$ the quantum bound $\mathcal{Q}$ limits the temporal Bell-like functional on the matrix of probability distributions, i.e. $B_{\tau}(A_i,A_{i+1})=\mathcal{F}([p_{t_i, t_{i+1}}(ab|xy)])\leq \mathcal{Q}$ with the association of histories $|H_{abxy})$, then the process saturating the chain for such n-steps can be designed in such a way that each pair of times is a replication of two consecutive times, i.e. $\forall_i [p_{t_i, t_{i+1}}]=[p_{t_{i+2}, t_{i+3}}]$. This process is equivalent logically to a loop $t_0\rightarrow t_1 \rightarrow t_0\rightarrow t_1\ldots$. In consequence, we get the following quantum bound :

\begin{equation}
\sum_{i=0}^{n-1} B_{\tau}(A_i,A_{i+1})\leq \mathcal{Q}n
\end{equation}
that can be saturated to its maximal value.
As an implication for the LGIs one gets the following quantum bound:
\begin{equation}
    \sum_{i=1}^n S_{LGI}(A_0,A_i)\leq 2\sqrt{2}n
\end{equation}
which can be saturated to the maximal value in quantum world and which violates the spatial monogamy relations $ \sum_{i=1}^n B(A_0,A_i)\leq 2n$ (for $n\geq 2$) \cite{Pawlowski}.

We can conclude this section with a remark that a particular entangled history is monogamous but for a bundle of histories with the same pre-selected and post-selected stated one can get violation of monogamy. This is a novel feature of temporal correlations not paralleled in spatial domain.

\section{Polyamory of an ensemble of quantum histories}

One should take the viewpoint that the present paper treats about entangled histories, however, the mathematical concepts related to temporal correlations seem to play a predominant role in our interpretation of spatio-temporal correlations. There are other representations like the multi-time state formalism (MSVF), process matrices or pseudo-density matrices and super-density formalism which gain a lot of attention. In general, as proved in \cite{NowakowskiCohen} for MSVF and entangled histories, they lead to the same results but there are subtle differences in what they represent. This sometimes leads to confusion in interpretation of the results.

In this paper, we have proved that a particular history is monogamous that leads to the well-known Tsirelson bound on the LGI and the bound for multi-system settings.
Yet the problem of the lack of monogamy e.g. in the evaporating space-times is still a field of an active research \cite{Hawking, Page}. It is suggested \cite{Grudka} that the temporal correlations in time are rather polyamorous and the lack of monogamy emerging in evaporating space-times is naturally related to lack of monogamy of correlations of outputs of measurements performed at subsequent instances of time of a single system.

As a particularly important example of potential polyamory of the temporal correlations, we can reconsider the case described in \cite{Grudka}. The evolution of the system
is such that at some point of time $t_{0}$ the single system can be viewed in an arbitrary chosen reference frame as a tripartite system $HAB$ being in a pre-selected state:
\begin{equation}
|\Psi_{0}\rangle=|\psi\rangle\otimes|\Phi^{+}\rangle
\end{equation}
 where $H$ is in a definite state $|\psi\rangle$ and $AB$ are projected onto the maximally entangled state $|\Phi^{+}\rangle$.
Then at time $t_{1}$ the particles $H$ and $A$ are always post-selected onto the state $|\Phi^{+}\rangle$:
\begin{equation}
|\Psi_{1}\rangle=|\Phi^{+}\rangle\otimes|\psi\rangle\
\end{equation}
Hence, it is concluded that the particle $A$ is maximally entangled with $H$ and with particle $B$ violating monogamy of entanglement which is a form of polyamory in time.

Yet, when we look at this interesting case by a prism of results of this paper, we find out that there is no dichotomy in viewing the monogamy of entanglement in time.
We shall write a two-time history of the system $HAB$ as:
\begin{equation}
|H)=[\Psi_{1}]\odot[\Psi_{0}]
\end{equation}
where one finds the particle $A$ maximally correlated with the particle $B$ in one reference frame at time $t_{0}$ and maximally correlated with the particle $H$ at time $t_{1}$. We need to emphasize that if we take a careful look at $A$ in the history $|H)$, it lives in a doubled space $\mathcal{B}(H_{A})\odot\mathcal{B}(H_{A})$
(one in time $t_{0}$ and one in time $t_{1}$).

We can consider an alternative interesting approach to address this case but this time explaining violation of monogamy in time by means of a mixture of entangled histories. One takes a perspective of a single system evolving from $H$ being prepared in a maximally mixed state $\rho_{H}=\frac{\mathbb{I}}{N}$ ($N=\dim \mathcal{B}(H_{A})$).
Then a set of consecutive measurements is performed: first we measure an observable $X=\sigma_x$, then $Y=\sigma_y$ and finally $Z=\sigma_y$.
It is straightforward to show that for consecutive measurements of X and Y on a mixed state, one gets the joint probability distribution $P(x,y|X,Y)$ in similarity to measurements of X and Y performed on a maximally entangled spatial state of a bipartite system
$|\Phi_{+}\rangle = \frac{1}{\sqrt{2}}(|00\rangle + |11\rangle)$. The same reasoning applies further to $p(y,z|Y,Z)$.
We conclude that these correlations violate temporal Bell-like inequalities and in consequence, we could expect also violation of
monogamy of quantum entanglement in time since we consider a single system.

This time we operate with an ensemble of entangled histories leading to violation of monogamous Bell-like temporal inequalities.
The action of the unitary Pauli operations (Fig 3.) on the subsystem A in between the pre-selected $|\Phi_{+}\rangle$ and the post-selected $|\Phi_{+}\rangle$ can be represented as a history $|H_{\Gamma})$:
\begin{equation}\label{Pauli}
|H_{\Gamma})=[\Phi_{+}]\odot[\Phi_{+}]\odot[\Phi_{+}]\odot[\Phi_{+}]\odot[\Phi_{+}]
\end{equation}
with bridging operators $U(t_{1},t_{2})=\sigma_{x}\otimes I_{B}$, $U(t_{2},t_{3})=\sigma_{y}\otimes I_{B}$, $U(t_{3},t_{4})=\sigma_{z}\otimes I_{B}$
and $U(t_{4},t_{5})=I$ with post-selected state $|\Phi_{+}\rangle$.
Interestingly, if we ask now for a history of the subsystem $A$ in this evolution, we get a temporal version of an entangled GHZ state:
\begin{equation}
|H_{A})=\frac{1}{2}([0]\odot[0]\odot[0]\odot[0]\odot[0]+[1]\odot[1]\odot[1]\odot[1]\odot[1])
\end{equation}
with corresponding evolution $U_{A}(t_{1},t_{2})=\sigma_{x}$, $U_{A}(t_{2},t_{3})=\sigma_{y}$, $U_{A}(t_{3},t_{4})=\sigma_{z}$
and $U_{A}(t_{4},t_{5})=I$
where $|H_{A})$ is derived from the global $|H_{\Gamma})$ tracing out $B$ party over all times (\cite{NowakowskiCohen}) keeping consistency of the derived reduced evolution of the subsystem (note that it is not a mere analogy of spatial trace out operation over all times, the binding evolution between the time steps has to be kept).

\begin{figure}[h]
\centerline{\includegraphics[width=9cm]{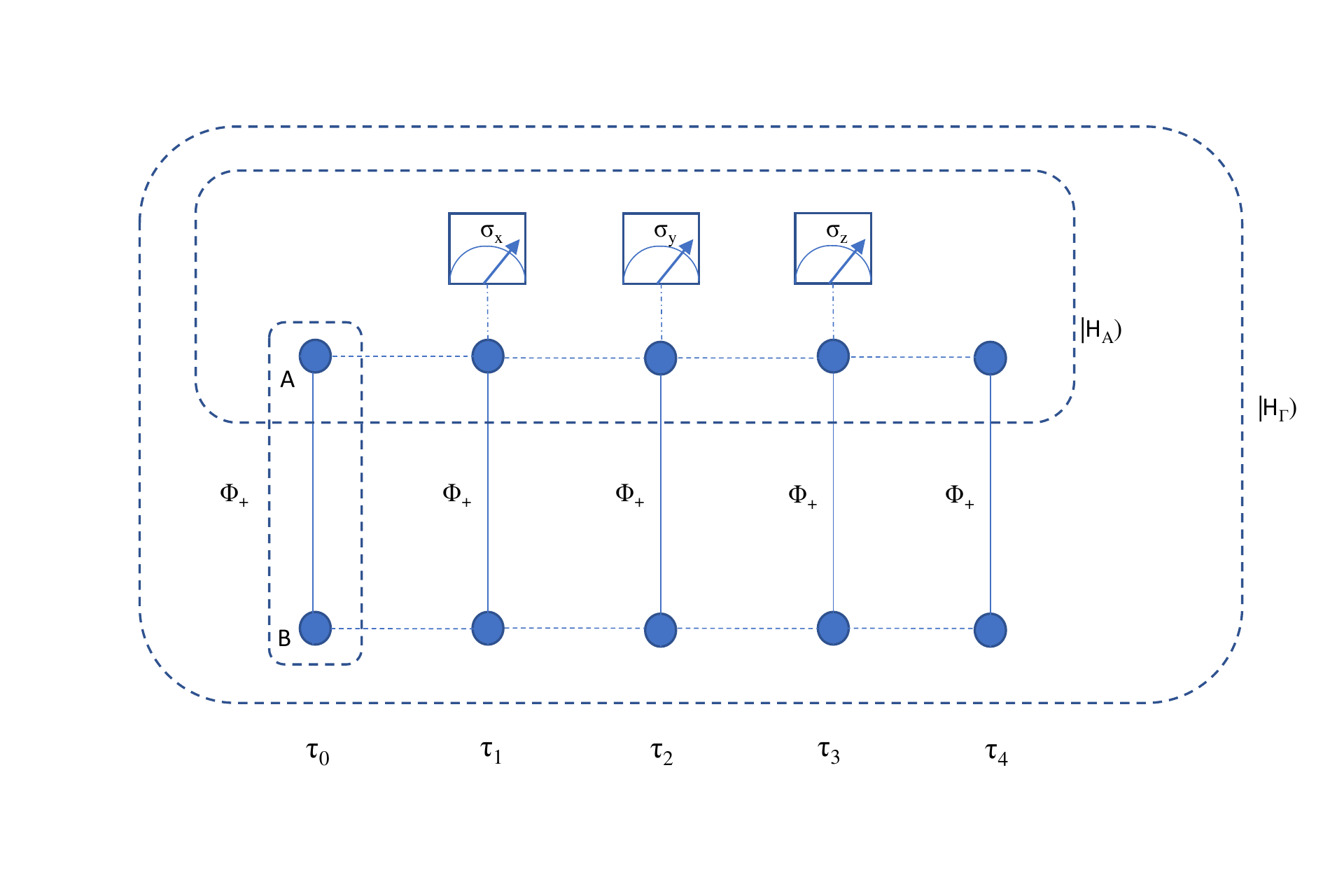}}
 \caption[Evolution]{$|H_{\Gamma})$ represents multi-time global evolution of the bipartite system $AB$ pre-selected and post-selected in state $\Phi_+$ with consecutive measurements on the subsystem $A$. One gets entangled temporal version of GHZ state for a history $|H_A)$ from a global separable history $|H_{\Gamma})$ which is a phenomenon unparalleled for spatial entanglement.  }
 \label{Global History}
\end{figure}

This local evolution can be represented equivalently with the following underlying quantum structure with trivial evolution:
\begin{eqnarray}\label{localev}
|H_{A})&=&\frac{1}{2}([0]\odot\sigma_{z}[0]\sigma_{z}^{\dag}\odot\sigma_{y}[0]\sigma_{y}^{\dag}\odot\sigma_{x}[0]\sigma_{x}^{\dag}\odot[0]\nonumber\\
&+&[1]\odot\sigma_{z}[1]\sigma_{z}^{\dag}\odot\sigma_{y}[1]\sigma_{y}^{\dag}\odot\sigma_{x}[1]\sigma_{x}^{\dag}\odot[1])
\end{eqnarray}

The above global and local history is constructed without a definite state of a local system A at intermediary instances of time. However,
if one reads the local state of A and gets a definite state $|+\rangle$ or $|-\rangle$ for a given measurement setting, then the history is disentangled and
effectively the considered history for the process is separable. Noticeably, one cannot forget that measurement itself is an inherent
element of the process and co-creates the particular history.  \\
As an example, consider probability  for the joint result $P(+++|XYZ)$ on the $A$ subsystem in evolution represented by $|H_{A})$ but with consecutive measurement settings $\sigma^+_x$, $\sigma^+_y$ and $\sigma^+_z$ :
\begin{equation}
P(+++|XYZ)=\frac{1}{2}|\langle 0|\sigma^+_z\sigma^+_y\sigma^+_x|0 \rangle+\langle 1|\sigma^+_z\sigma^+_y\sigma^+_x| 1\rangle|^2
\end{equation}
due to cancellation of the zero-probability evolution elements.
Thus, for a history:
$|H_{A})=\frac{1}{4}I\odot[z^{+}]\odot[y^{+}]\odot[x^{+}]\odot I $ (one puts $z^+=\frac{I+\sigma_{z}}{2}$, $x^+=\frac{I+\sigma_{x}}{2}$ etc.) we get the realization
probability $Pr(|H_{A}))=Tr(K(|H_{A}))K(|H_{A})))$.
It is worth mentioning that that the same result can be derived applying the ABL formula  \cite{ABL} as discussed in the previous section.

It is to be emphasized that if we measure e.g. $P(00|XY)$, there is no entanglement for this two-point function since we get a direct result of $x=0$ for $\sigma_x$
and $y=0$ for $\sigma_y$ disentangling the particular history. There is also no violation of monogamy of quantum entanglement in time. Conversely, if we impose violation of monogamy of quantum entanglement in time, then the physical structure of the probability amplitude of the underlying quantum process should correspond to that. Yet, as proved above, we rather get a temporal version of the $GHZ$ state for a multi-time process and violation of monogamy for the ensemble of histories with different measurement results for (\ref{localev}) corresponding to the whole probability distribution $P(xyz|XYZ)$.

There is one more important aspect related to calculation of statistics for consecutive measurements of the single system. One cannot assume in calculation of $P(xy|XY)$ and $P(yz|YZ)$ that an observer can make the consecutive measurements more or less independently. Such an assumption  imposed in analogy to the measurements made on ensembles of quantum pairs for the spatial entanglement would lead to incorrect results. For the sake of temporal correlations of a single system, measurements themselves are part of the particular histories (evolution). Thus, if we get for the intermediate step e.g. $P(xy=0|XY)$, we need to keep $y=0$ for further evolution of the particular history and further calculation of $P(yz|YZ)$ with $y=0$ even if we take the ensemble of evolutions for the last step.


We need to remind that the initial and final conditions on the evolution of the system imply that not all intermediate measurement results are possible as we count only non-zero probability amplitudes.
It is pointed out that what mimics violation of monogamy of entanglement is actually just a kind of polyamory in time but monogamy of entanglement for a particular evolution still holds.

\section{Conclusions}

The central idea of this paper was to show how the quantum correlations in time can violate Bell-like monogamous inequalities conserving monogamy for particular processes.

We proved that a particular entangled history, which can be associated with a quantum propagator, is monogamous to conserve its consistency throughout time. However, the evolving systems can still violate monogamous Bell-like multi-time inequalities. This dichotomy, being a novel feature of temporal correlations, has its roots in the measurement process which is discussed by means of the bundles of entangled histories. The measurement process is an inherent part of the quantum evolution and different measurement outcomes generate different instances of the considered evolution. We introduced and discussed a concept of a probabilistic mixture of quantum processes by means of which we clarify why the spatial-like Bell-type monogamous inequalities are further violated.
We derived the quantum bound for multi-time Bell-like inequalities basing on the Tsirelson bound on temporal Bell-like inequalities derived from the entangled histories approach. This result is interesting due to the fact that previous methods were based on the linear optimization on the set of allowed probabilistic distributions generated by quantum measurements.
In the context of the black hole information paradox, it is also pointed out that what mimics violation of monogamy of temporal entanglement is actually just a kind of polyamory in time but monogamy of entanglement for a particular evolution still holds. We employed also a novel feature of temporal correlations which seem to be separable for a global evolution generating locally entangled temporal states to violate monogamous Bell-like inequalities in space-time.


There are many open problems and questions for further research in this field. Future research can be focused on analysis of non-locality in time and finding more appropriate mathematical structures that will enable easier calculations of measurements' outputs for observers in different reference frames. Monogamy of entanglement in time and non-locality in time can be probably applied also in quantum cryptography and should give some new insights into non-sequential quantum algorithms and information processing. Finally, as stated in the paper the subject is fundamental for understanding relativistic quantum information theory and brings new prospects for this field.

\section{Acknowledgments}
Acknowledgments to Eliahu Cohen, Pawel Horodecki and Jan Tuziemski for discussions about temporal correlations. Part of this work was performed at the National Quantum Information Center in Gdansk.


\begin{thebibliography}{8}
\bibitem{TraceTime}
For a history $|Y_{t_{n}\ldots t_{0}})=F_{n}\odot F_{n-1}\odot\ldots\odot F_{1}\odot F_{0}$ acting on a space $\mathcal{H}=\mathcal{H}_{t_{n}}\otimes\dots\otimes\mathcal{H}_{t_{0}}$, a partial trace over times $\{t_{j}\ldots t_{i+1} t_{i}\}$ $(j\geq i)$ is:
$
Tr_{t_{j}\ldots t_{i+1}t_{i}} |Y_{t_{n}\ldots t_{0}})(Y_{t_{n}\ldots t_{0}}|=\sum_{k=1}^{\dim\mathcal{F}} (e_{k}|Y_{t_{n}\ldots t_{0}})(Y_{t_{n}\ldots t_{0}}|e_{k})
$
where $\mathcal{F}=\{|e_{k})\}$ creates an orthonormal consistent family of histories on times $\{t_{j}\ldots t_{i+1} t_{i}\}$ and the strong consistency condition for partial histories holds for base histories, i.e. $(e_{i}|e_{j})=Tr[K(|e_{i}))^{\dag}K(|e_{j}))]=\delta_{ij}$.
\end{thebibliography}
\end{document}